\documentclass[pra,amsmath,amsfonts,twocolumn,superscriptaddress,showpacs]{revtex4-1}
 \usepackage{epsf,epsfig}
\usepackage{amsthm}

\newtheorem{theorem}{Theorem}

\newtheorem{remark}{Remark}
\newtheorem{observation}{Observation}
 \newcommand{\ket}[1]{|#1\rangle}
 \newcommand{\bra}[1]{\langle #1|}
 \newcommand{\project}[1]{\ket{#1}\bra{#1}}

 \newcommand{\Id}{{\mathbb I}}
 \newcommand{\e}{{\mathrm e}}

 \newcommand{\T}{{\mathrm t}}

\newcommand{\dif}{{\mathrm d}}
\newcommand{\Tr}{{\mathrm {Tr}}}
\newcommand{\diag}{{\mathrm {diag}}}

\begin{document}
\title{Progress on Quantum Discord of  Two-Qubit States: Optimization and Upper Bound}

\author{S. Javad Akhtarshenas\footnote{akhtarshenas@phys.ui.ac.ir}}
\affiliation{Department of Physics, University of Isfahan,
 Isfahan, Iran}
\affiliation{Quantum Optics Group, University of Isfahan,
 Isfahan, Iran}
 \affiliation{Department of Physics, Ferdowsi University of Mashhad,
 Mashhad, Iran}
\author{Hamidreza Mohammadi\footnote{h.mohammadi@sci.ui.ac.ir}}\affiliation{Department of Physics, University of Isfahan,
 Isfahan, Iran}
\affiliation{Quantum Optics Group, University of Isfahan,
 Isfahan, Iran}
\author{Fahimeh S. Mousavi\footnote{fahimeh_mousavi_90@sci.ui.ac.ir}}\affiliation{Department of Physics, University of Isfahan,
 Isfahan, Iran}
\author{Vahid Nassajpour\footnote{v.nasaj@yahoo.com}}
\affiliation{Department of Physics, University of Isfahan,
 Isfahan, Iran}

\begin{abstract}
Calculation of the quantum discord requires to find the minimum of the quantum conditional entropy $S(\rho^{AB}|\{\Pi^B_{k}\})$ over all measurements on the subsystem  $B$. In this paper, we provide a simple relation for the conditional entropy as the difference of two Shannon entropies. The relation is suitable for calculation of the quantum discord in the sense that it can be used to  obtain the quantum discord for some classes of two-qubit states such as Bell-diagonal states and a three-parameter subclass of X states, without the need for    minimization.   We also present an analytical procedure of optimization and obtain conditions under which the quantum conditional entropy of a general two-qubit state is stationary. The presented relation is also used to find a tight upper bound on the quantum discord.
\end{abstract}


\maketitle
\section{INTRODUCTION}

Entanglement is the specific feature of quantum systems which reveals that complete information about parts of a composite system does not include complete information of the whole system \cite{EPR,S}. However, this is not the only weird character of a quantum system. For instance, the collapse of one part of a non-entangled bipartite quantum system after a measurement on the other part is another feature unique to quantum systems. The quantity that captures this feature is the quantum discord \cite{ZurekPRL2001,HendersonJPA2001}. The key idea of the concept of quantum discord is the superposition principle and the vanishing of discord shown to be a criterion for quantum to classical transition \cite{ZurekPRL2001}. Furthermore, quantum discord has a simple interpretation in thermodynamics and has been used in analyzing the power of a quantum Maxwell's demon \cite{ZurekPRA2003}. It has also been employed in the study of pure quantum states as a resource and  performance of deterministic quantum computation with one pure qubit \cite{KnillPRL1998}. The authors of Refs. \cite{DattaPRL2008,LanyonPRL2008} showed that non-classical correlations other than entanglement can be responsible for the quantum computational efficiency of deterministic quantum computation with one pure qubit \cite{KnillPRL1998} and brought this obscured correlation measure to the spot light zone. After this discovery, quantum discord becomes one of the most frequent topics of researches in the field of quantum information theory.  Indeed, quantum discord is the difference between two classically equivalent definitions of mutual information in the quantum mechanics language.
In mathematical sense, discord could be obtained by eliminating the whole of classical correlation from the total correlation measured by mutual information, by means of the most destructive measurement on the one party of the system. Mutual information of a bipartite system  can be written as
\begin{equation}\label{MI}
I(\rho^{AB})=S(\rho^{A})+S(\rho^{B})-S(\rho^{AB}),
\end{equation}
where $\rho^{A}$ and $\rho^{B}$ refer to the reduced density
matrices of the subsystems  $A$ and $B$, respectively, $\rho^{AB}$ is the density matrix of the system as
the whole, and $S(\rho)=-\Tr(\rho\,\log_{2}\rho)$ is the Von Neumann entropy.
The classical correlation between the parts of a bipartite system
can be obtained by use of the measurement-base conditional density
operator and can be written as \cite{ZurekPRL2001}
\begin{equation}\label{CC}
C_{B}(\rho^{AB})=\underset{\{\Pi^B_{k}\}}{\sup}\{S(\rho^{A})-S(\rho^{A}|\{\Pi^B_{k}\})\}.
\end{equation}
Here the maximum is taken over all projective measurement $\{\Pi^B_{k}\}$ on the subsystem $B$ \cite{ZurekPRL2001},
and $S(\rho^{A}|\{\Pi^B_{k}\})=\sum_{k}p_{k}S(\rho^A_{k})$ is
the conditional entropy of the subsystem $A$, with $\rho^A_{k}=\Tr_{B}((\Id^{A}\otimes\Pi^B_{k})\,\:\rho^{AB}\,\:(\Id^{A}\otimes\Pi^B_{k}))/p_{k}$ as the post-measurement state of the subsystem $A$
and $p_{k}=\Tr(\rho^{AB}\:(\Id^{A}\otimes\Pi^B_{k}))$ being the probability of the $k$-th outcome.  The maximum performed in Eq. (\ref{CC}) can  be taken also over all the positive operator valued measures (POVM) \cite{HendersonJPA2001} and these two definitions  give in general inequivalent results.
Accordingly, discord
can be calculated as follows
\begin{equation}\label{Discord}
D_{B}(\rho^{AB})=I(\rho^{AB})-C_{B}(\rho^{AB}).
\end{equation}
However, one can swap the role of the subsystems $A$ and $B$ to
obtain $D_{A}(\rho^{AB})$, which is not equal to $D_{B}(\rho^{AB})$ in general. In this paper we only consider $D_{B}(\rho^{AB})$
and hence ignore the subscript $B$ in the following.

The optimization procedure involved  in the calculation of quantum discord  prevents one to write an analytical expression for quantum discord  even for simple two-qubit systems. Quantum discord  is analytically computed only for a few families of states including the Bell-diagonal states \cite{LuoPRA2008,LangPRL2010}, two-qubit $X$ states \cite{AliPRA2010,ChenPRA2011}, two-qubit rank-2 states \cite{ShiJPA2011}, a class of rank-2 states of $4\otimes 2$ systems \cite{CenPRA2011}, and Gaussian states of the continuous variable systems \cite{AdessoPRL2010}. Moreover, based on the optimization of the conditional entropy, an algorithm to calculate the quantum discord of the two-qubit states  is presented in \cite{GirolamiPRA2011}. It is also important to have some computable bounds on the quantum discord and some authors have obtained such bounds \cite{XiJPA2011,YuQuant1102.1301}.

In this paper, we consider two-qubit states and obtain a simple relation for the quantum conditional entropy $S(\rho^{AB}|\{\Pi^B_{k}\})$ as the difference of two Shannon entropies. This form of the conditional entropy enables one to calculate its minimum for some classes of two-qubit states such as Bell-diagonal states \cite{LuoPRA2008} and a three-parameter subclass of X states \cite{AliPRA2010,ChenPRA2011}, without any  minimization procedure. Although, the quantum discord of these states is already obtained analytically, but the presented form for the conditional entropy enables one to obtain the previous results in a much simpler  manner.
 An analytical progress in the minimization of the conditional entropy of a general two-qubit state is also presented. Our algorithm presents a necessary and sufficient condition for a measurement to be the stationary measurement for the conditional entropy. The presented condition is, to the best of our knowledg, more efficient relative to the earlier presented optimization algorithms.    Moreover, we obtain  a computable tight upper bound on the quantum discord of an arbitrary two-qubit state. We also present sufficient conditions under which the upper bound is tigh, and  exemplify  this bound for  a two-parameter class of states and show that the bound may be tight even in the absence of such sufficient conditions.

The paper is organized as follows. In section II,  we consider a general two-qubit system and present a simple relation for the quantum conditional entropy. In this section, we also evaluate the quantum discord for some classes of states. In section III, we present a tight upper bound on the discord of a general two-qubit state. Section IV is devoted to the optimization procedure. An analytical conditions under which the conditional entropy is stationary is presented in this section. The paper is concluded in section V with a brief discussion.

\section{Conditional entropy}
A general two-qubit state can be written in the Hilbert-Schmidt representation as
\begin{equation}\label{RhoTwoQubit}
\rho^{AB}=\frac{1}{4}\left(\Id\otimes \Id+\vec{x}\cdot{\mathbf \sigma}\otimes \Id+\Id\otimes {\vec y}\cdot{\mathbf \sigma}+\sum_{i,j=1}^3t_{ij} \sigma_i\otimes\sigma_j\right).
\end{equation}
Here $\Id$ stands for the identity operator,  $\{\sigma_i\}_{i=1}^3$ are the standard Pauli matrices,  $\vec{x}$ and $\vec{y}$ are coherence vectors of the subsystems $A$ and $B$, respectively,  and $T=(t_{ij})$ is the correlation matrix. Therefore, to each state $\rho^{AB}$ we associate the triple $\{\vec{x},\vec{y},T\}$.
Since quantum correlations are invariant under local unitary transformation, i.e. under transformations of the form $(U_1\otimes U_2)\rho^{AB}(U_1^\dag\otimes U_2^\dag)$ with $U_1,U_2\in SU(2)$, we can, without loss of generality, restrict our considerations
to some representative class of states described
by less number of parameters \cite{HorodeckiPRA1996}. Under such transformations, the triple $\{\vec{x},\vec{y},T\}$ transforms as
\begin{equation}\label{O1O2}
{\vec x}\rightarrow O_1{\vec x},\qquad {\vec y}\rightarrow O_2{\vec y},\qquad T\rightarrow O_1TO_2^{\T},
\end{equation}
where $O_i$'s corresponds to $U_i$'s via $U_i({\vec a}\cdot {\vec \sigma})U_i^\dag= (O_i{\vec a})\cdot {\vec \sigma}$  \cite{HorodeckiPRA1996}.
 In view of this, any state of the two-qubit system can be written as $(U_1\otimes U_2)\rho^{AB}(U_1^\dag\otimes U_2^\dag)$, where $\rho^{AB}$ belongs to the representative class.
In the following we will consider a representative class such that $T$ is diagonal, namely $T=\diag\{t_1,t_2,t_3\}$.
Concerning this representative class, a general state of two-qubit system can be parameterized by nine real parameters $\vec{x}=(x_1,x_2,x_3)^{\T}$, $\vec{y}=(y_1,y_2,y_3)^{\T}$,  and $T=\diag\{t_1,t_2,t_3\}$, where $\T$ denotes transposition.  Accordingly, in the computational basis $\{\ket{00},\ket{01},\ket{10},\ket{11}\}$, a general state $\rho^{AB}$ of this representative class takes the following form
\begin{equation}\label{RhoTwoQubitMatrix}
\rho^{AB}=\frac{1}{4}\left(\begin{array}{cccc}\rho_{11} & y_1-iy_2 & x_1-ix_2 & t_1-t_2 \\
y_1+iy_2 & \rho_{22} & t_1+t_2 & x_1-ix_2 \\
x_1+ix_2 & t_1+t_2 & \rho_{33} & y_1-iy_2 \\
t_1-t_2 & x_1+ix_2 & y_1+iy_2 & \rho_{44}
\end{array}\right),
\end{equation}
where
\begin{eqnarray}\nonumber
\rho_{11}=1+x_3+y_3+t_3, \quad \rho_{22}=1+x_3-y_3-t_3, \\ \nonumber
\rho_{33}=1-x_3+y_3-t_3, \quad  \rho_{44}=1-x_3-y_3+t_3.
\end{eqnarray}
Now let us turn our attention on the von Neumann measurement on the qubit B.
A general such measurement can be written as
\begin{equation}\label{ProjectiveUk}
\Pi_k^B=U\project{k} U^\dag,
\end{equation}
where  $\{\project{k}\}_{k=0}^1$ is the von Neumann projection operators in the standard basis of the qubit $B$, and $U\in SU(2)$. An arbitrary element of $SU(2)$ can be factorized as \cite{GilmoreBook1974}
\begin{equation}
U=\Omega_2\Omega_1,
\end{equation}
with $\Omega_2$ and $\Omega_1$ defined by
\begin{equation}
\Omega_2=\left(\begin{array}{cc}\cos{\frac{\theta}{2}} & -\e^{-i\phi}\sin{\frac{\theta}{2}} \\ \e^{i\phi}\sin{\frac{\theta}{2}} & \cos{\frac{\theta}{2}}\end{array}\right),
\;\; \Omega_1=\left(\begin{array}{cc} \e^{i\eta/2} & 0 \\ 0 & \e^{-i\eta/2} \end{array}\right),
\end{equation}
for $0\le \theta \le \pi$, $0\le \phi \le 2\pi$, and $0\le \eta \le \pi$.
Therefor,  we get $\Pi_k^B=\Omega_2 \project{k} \Omega_2^\dag=\project{\sigma\cdot {\hat n}_k}$ where ${\hat n}_0=-{\hat n}_1={\hat n}$, with ${\hat n}=(\sin{\theta}\cos{\phi}, \sin{\theta}\sin{\phi}, \cos{\theta})^{\T}$.   This, explicitly, shows that only two independent parameters  $\theta$ and $\phi$ are needed to characterize a general local von Neumann measurement on the two-qubit systems.  We can also write these orthogonal projections in the Bloch representation as
\begin{equation}\label{ProjectiveCV}
\Pi_k^B=\frac{1}{2}\left(\Id+{\hat n}_k\cdot \sigma\right), \qquad k=0,1.
\end{equation}
Therefore ${\hat n}_0$ and ${\hat n}_1$ are coherence vectors of $\Pi^B_0$ and $\Pi^B_1$, respectively.
For further use, we calculate the expression $\Pi_k^B\sigma_j \Pi_k^B$ for $k=0,1$ and $j=1,2,3$, i.e.
\begin{eqnarray}\nonumber
\Pi_k^B\sigma_j \Pi_k^B&=&\project{\sigma\cdot {\hat n}_k}\sigma_j\project{\sigma\cdot {\hat n}_k} \\ \label{PiSigmaPi}
&=& \Tr{\left(\Pi_k^B\sigma_j\right)}\Pi_k^B=(\hat{n}_k)_j \Pi_k^B.
\end{eqnarray}
 Now we are in the position to perform the von Neumann measurement $\{\Pi_k^B\}_{k=0}^1$  on the qubit $B$. This transforms the total state
 $\rho^{AB}$ to the ensemble $\{p_k, \rho^{AB}_k\}_{k=0}^{1}$ such that
 \begin{equation}\label{RhoAk-Definition}
 \rho^{AB}_k=\frac{1}{p_k}(\Id\otimes \Pi_k^B)\rho^{AB}(\Id\otimes \Pi_k^B),
 \end{equation}
 with $p_k=\Tr{[(\Id\otimes \Pi_k^B)\rho^{AB}(\Id\otimes \Pi_k^B)]}$. By using Eqs. (\ref{RhoTwoQubit}) and (\ref{ProjectiveCV}) in  (\ref{RhoAk-Definition}) and invoking relation (\ref{PiSigmaPi}) we get
 \begin{equation}
 \rho^{AB}_k=\rho^A_k\otimes \Pi_k^B,
 \end{equation}
 where
 \begin{equation}\label{RhoAk}
 \rho^A_k=\frac{1}{2}\left(\Id+\vec{{\tilde x}}_k\cdot\sigma\right),
 \end{equation}
 is the post-measurement state of the qubit $A$, associated to the measurement result $k$ with the  corresponding probability
 \begin{equation}\label{pk}
 p_k=\frac{1}{2}\left(1+\vec{y}^{\T} \hat{n}_k\right).
 \end{equation}
 In Eq. (\ref{RhoAk}), the post-measurement coherence vector $\vec{{\tilde x}}_k$ is defined by
 \begin{equation}
\vec{{\tilde x}}_k=\frac{\vec{x}+T{\hat n}_k}{1+\vec{y}^{\T} \hat{n}_k}.
 \end{equation}
 The quantum conditional entropy with respect to the above measurement is defined by
 \begin{equation}\label{QC-Entropy-1}
 S(\rho^A|\{\Pi_k^B\})=p_0 S(\rho^A_0|\{\Pi_k^B\})+p_1 S(\rho^A_1|\{\Pi_k^B\}).
 \end{equation}
 Now using
 \begin{equation}
 \frac{1}{2}\left(1\pm |\vec{{\tilde x}}_k|\right)= \frac{1}{4p_k}\left(2p_{k}\pm |\vec{x}+T{\hat n}_k|\right),
 \end{equation}
 as the eigenvalues of $\rho^A_{k}$, for $k=1,2$, and after some calculations we arrive at the following observation  for the conditional entropy.
 \begin{observation}
 Conditional entropy can be written as
  \begin{equation}\label{QC-Entropy-2}
  S(\rho^A|\{\Pi_k^B\})=h_4(\vec{w})-h_2(p_0),
 \end{equation}
where above, and throughout this paper, $h_2(x)$ denotes the binary Shannon entropy  \cite{NielsenBook2010} defined by
 \begin{equation}
 h_2(x)=-x\log_{2} x-(1-x)\log_{2}(1-x),
 \end{equation}
 and $h_{m}(q_1,\cdots,q_m)=-\sum_{i=1}^m q_i\log_{2}{q_i}$ is the Shannon entropy of the probabilities $\{q_1,\cdots,q_m\}$.
In particular,  $h_4(\vec{w})=-\sum_{i=1}^4 w_i\log_{2}{w_i}$ is  the Shannon entropy of the probabilities
 \begin{equation}\label{w1234}
 w_{1,2}=\frac{2p_0\pm |\vec{x}+T{\hat n}|}{4},\quad
 w_{3,4}=\frac{2p_1\pm |\vec{x}-T{\hat n}|}{4}.
 \end{equation}
 \end{observation}
 Note that under the transformation  ${\hat n}\rightarrow -{\hat n}$, corresponding to $\theta\rightarrow \pi-\theta$ and $\phi\rightarrow \phi\pm\pi$, the probabilities (\ref{pk}) and (\ref{w1234}) transform as $p_0\leftrightarrow p_1$, $w_1\leftrightarrow w_3$ and $w_2\leftrightarrow w_4$, leaving therefore the conditional entropy invariant.
 On the other hand,
 if we perform local unitary transformation (\ref{O1O2}) on the density matrix, the probabilities (\ref{pk}) and (\ref{w1234}), and hence the conditional entropy do not change provided we perform the transformation ${\hat n}\rightarrow O_2{\hat n}$. This implies that if we find $\hat{n}^\ast$ as the optimal measurement for a given state $\rho^{AB}$ of the representative class, one can obtain the optimal one for any state $\tilde{\rho}^{AB}=(U_1\otimes U_2)\rho^{AB}(U_1^\dag\otimes U_2^\dag)$, just by the transformation ${\hat n}^\ast\rightarrow O_2{\hat n}^\ast$.

Now the aim is to minimize the above conditional entropy.
Before we give a general procedure for optimization, we give below some special classes of states for which the quantum discord can be evaluated without the need for any optimization.

\noindent {\bf Quantum Discord of   States with $T^{\T}\vec{x}=0$ and $\vec{y}=0$.}
Let us consider a three-parameter class of states such that $\vec{y}=0$ and $\vec{x}$ belongs to the kernel of $T^{\T}$, i.e. $T^{\T}\vec{x}=0$.
For this class of states we get
\begin{equation}\label{x+Tn=x-Tn}
|x+T{\hat n}|=|x-T{\hat n}|=\sqrt{x^2+{{\hat n}^{\T}}T^{\T}T{\hat n}},
\end{equation}
and therefore
\begin{equation}
p_0=p_1=\frac{1}{2},\qquad w_{1,2}=w_{3,4}=\frac{1}{4}\left(1\pm|{\vec x}+T{\hat n}|\right).
\end{equation}
In this case we get
\begin{equation}
h_4(\vec{w})=h_2\left(\frac{1+|x+T{\hat n}|}{2}\right)+1,\qquad h_2(p_0)=1,
\end{equation}
and therefore Eq. (\ref{QC-Entropy-2}) reduces to
\begin{equation}
S(\rho^A|\{\hat{n}\})=h_2\left(\frac{1+|x+T{\hat n}|}{2}\right).
\end{equation}
Clearly, the minimum of the above equation occurs whenever  $|x+T{\hat n}|=\sqrt{x^2+{\hat n}^{\T}T^{\T}T{\hat n}}$ takes its maximum value. This happens when ${\hat n}$ is an eigenvector of $T^{\T}T$ corresponding to the largest eigenvalue $t_{\max}^2$, therefore
\begin{equation}
\min S(\rho^A|\{{\hat n}\})=h_2\left(\frac{1+\sqrt{x^2+t^2_{\max}}}{2}\right).
\end{equation}
For these states quantum discord is
\begin{eqnarray}\nonumber
D(\rho^{AB})=1-h_4(\mu_1,\mu_2,\mu_3,\mu_4)+h_2\left(\frac{1+\sqrt{x^2+t^2_{\max}}}{2}\right),
\end{eqnarray}
where $\{\mu_i\}_{i=1}^{4}$ are eigenvalues of $\rho^{AB}$.

Note that if we concern the representative class of states (\ref{RhoTwoQubitMatrix}) for which $T=\diag\{t_1,t_2,t_3\}$, then condition $T^{\T}\vec{x}=0$ requires that $t_ix_i=0$ for $i=1,2,3$. Hence if we take, without loss of generality,  $|t_1|\ge |t_2|\ge |t_3|\ge0$ and set $\vec{y}=0$,  then the states corresponding to this class can be obtained from the general form of Eq. (\ref{RhoTwoQubitMatrix}) as: \\
(i) {\it States with $x_1=x_2=x_3=0$.---}  This corresponds to the Bell-diagonal states.
In this case discord reads as
\begin{eqnarray}\nonumber
D(\rho^{AB})=1-h_4(\mu_1,\mu_2,\mu_3,\mu_4)+h_2\left(\frac{1+|t_1|}{2}\right),
\end{eqnarray}
with $\{\mu_i\}_{i=1}^{4}$ as eigenvalues of $\rho^{AB}$ given by
\begin{eqnarray}\nonumber
\mu_{1,2}=\frac{1}{4}\left(1\pm t_1\pm t_2-t_3\right), \quad
\mu_{3,4}=\frac{1}{4}\left(1\pm t_1\mp t_2+t_3\right).
\end{eqnarray}
This is in agreement with the result obtained by Luo in \cite{LuoPRA2008} (see also \cite{LangPRL2010}).
\\
(ii) {\it States with $x_1=x_2=t_3=0$.---} This corresponds to a three-parameter subclass of the so-called $X$ states.  In this case, the discord is obtained as
\begin{eqnarray}\nonumber
D(\rho^{AB})=1-h_4(\mu_1,\mu_2,\mu_3,\mu_4)+h_2\left(\frac{1+\sqrt{t_1^2+x_3^2}}{2}\right),
\end{eqnarray}
where $\{\mu_i\}_{i=1}^{4}$ are eigenvalues of $\rho^{AB}$ given by
\begin{eqnarray}\nonumber
\mu_{1,2}&=&\frac{1}{4}\left(1\pm\sqrt{(t_1+t_2)^2+x_3^2}\right), \\ \nonumber
\mu_{3,4}&=&\frac{1}{4}\left(1\pm\sqrt{(t_1-t_2)^2+x_3^2}\right).
\end{eqnarray}
(iii) {\it States with  $x_1=t_2=t_3=0$.---} This corresponds to a three-parameter subclass of the zero-discord states. \\
(iv) {\it States with  $t_1=t_2=t_3=0$.---} This also corresponds to a three-parameter subclass of the zero-discord states. \\

\section{Tight upper bound of quantum discord}
Interestingly, the above examples motivate us to  introduce an upper bound for the quantum discord.
Let $\mathcal{R}$ be the subspace of $\mathbb{R}^3$ spanned by $T^{\T}\vec{x}$ and $\vec{y}$, i.e.  $\mathcal{R}=\textrm{span}\{T^{\T}\vec{x},\vec{y}\}$, and let  $\mathcal{R}^{\perp}$ denotes the orthogonal complement of  $\mathcal{R}$, i.e. the set of all vectors in $\mathbb{R}^3$ that are orthogonal to every element of $\mathcal{R}$. Hence we have  $\mathcal{R}+\mathcal{R}^{\perp}=\mathbb{R}^3$.
\begin{theorem}\label{UpperBoundT}
The conditional entropy (\ref{QC-Entropy-2}) is bounded from above as
\begin{eqnarray}\label{CE-UpperBound}
\min_{\{\Pi_k^B\}}{S(\rho^A|\{\Pi_k^B\})}\le  h_2\left(\frac{1+\sqrt{x^2+t_{0}^2}}{2}\right),
\end{eqnarray}
where $x=|\vec{x}|$, and
\begin{equation}\label{t02}
t_{0}^2=\max_{\hat{e}_0\in \mathcal{R}^{\perp}}\hat{e}_0^{\T}T^{\T}T\hat{e}_0.
\end{equation}
Accordingly, the classical correlation and the quantum discord have the following lower and upper bounds, respectively
\begin{eqnarray}\label{C-LowerBound}
C(\rho^{AB})&\ge & S(\rho^A)-h_2\left(\frac{1+\sqrt{x^2+t_{0}^2}}{2}\right), \\ \label{D-UpperBound}
Q(\rho^{AB})& \le & S(\rho^B)-S(\rho^{AB})+ h_2\left(\frac{1+\sqrt{x^2+t_{0}^2}}{2}\right).
\end{eqnarray}
\end{theorem}
\begin{proof}
Let us concern about all measurement vectors  $\hat{e}_0$ living in  $\mathcal{R}^{\perp}$, i.e.
$\vec{y}^{\T}\hat{e}_0=0$, $(T^{\T} \vec{x})^{\T}\hat{e}_0=0$; then
\begin{equation}
|x+T{\hat e}_0|=|x-T{\hat e}_0|=\sqrt{x^2+{{\hat e}_0^{\T}}T^{\T}T{\hat e}_0},
\end{equation}
and
\begin{equation}
p_0=p_1=\frac{1}{2},\quad w_{1,2}=w_{3,4}=\frac{1}{4}\left(1\pm|{\vec x}+T{\hat e}_0|\right).
\end{equation}
Therefore
\begin{eqnarray}\nonumber
\min_{\hat{n}\in \mathbb{R}^3}{S(\rho^A|\{\hat{n}\})} &\le & \min_{\hat{e}_0\in \mathcal{R}^{\perp}}S(\rho^A|\{\hat{e}_0\}) \\ \nonumber
&=& \min_{\hat{e}_0\in \mathcal{R}^{\perp}}h_2\left(\frac{1+\sqrt{x^2+{\hat e}_0^{\T} T^{\T}T{\hat e}_0}}{2}\right) \\ \label{t0}
&=& h_2\left(\frac{1+\sqrt{x^2+t_{0}^2}}{2}\right),
\end{eqnarray}
where $t_{0}^2=\max\hat{e}_0^{\T}T^{\T}T\hat{e}_0$ with the maximum  taken over all unit vectors $\hat{e}_0\in \mathcal{R}^{\perp}$.
Evidently, if two vectors $T^{\T}\vec{x}$ and $\vec{y}$ be nonzero and linearly  independent then $\dim{\mathcal{R}^{\perp}}=1$, so that the unit vector $\hat{e}_0\in\mathcal{R}^{\perp}$ is unique. Otherwise $\dim{\mathcal{R}^{\perp}}>1$, so  we can choose $\hat{e}_0\in \mathcal{R}^{\perp}$ such that  $t_{0}^2=\underset{\hat{e}_0\in \mathcal{R}^{\perp}}{\max}\hat{e}_0^{\T}T^{\T}T\hat{e}_0$, giving a tighter bound.  This completes the proof of (\ref{CE-UpperBound}).  Using Eq. (\ref{CE-UpperBound})   in Eqs.  (\ref{CC}) and (\ref{Discord}), we obtain the desired bounds (\ref{C-LowerBound}) and (\ref{D-UpperBound}), respectively.
\end{proof}

Remarkably, the above bound is tight in the sense that  for all states that ${\mathcal{R}^{\perp}}=\mathbb{R}^3$, the bound is achieved.  This happens when  $T^{\T}\vec{x}=\vec{y}=0$,  gives therefore a sufficient condition for the reachable upper bound of the quantum discord. In this  case $t_0^2$ becomes the largest eigenvalue of $T^{\T}T$ and $\hat{e}_0$ is the corresponding eigenvector. Surprisingly, as we will show in the example below, it may happens  ${\mathcal{R}^{\perp}}\ne\mathbb{R}^3$ but the upper bound (\ref{CE-UpperBound}) is achieved.  In such  cases the absolute minimum of the conditional entropy happens for some vectors in $\mathcal{R}^{\perp}\subseteq\mathbb{R}^3$.
Recently an upper bound for the quantum discord is obtained in \cite{XiJPA2011} as $Q(\rho^{AB}) \le  S(\rho^B)$. A comparison of this with the upper bound presented in (\ref{D-UpperBound}) shows that for all states  for which  $S(\rho^{AB})-h_2\left(\frac{1+\sqrt{x^2+t_0^2}}{2}\right)>0$, our bound is stronger.

As an illustrative example let us consider a two-parameter class of states discussed in  \cite{AlQasimiPRA2011}
\begin{eqnarray}\label{Rho-ab}
\rho^{AB}(a,b)=\frac{1}{2}\left(\begin{array}{cccc}a & 0 & 0 & a \\ 0 & 1-a-b & 0 & 0 \\ 0 & 0 & 1-a+b & 0 \\ a & 0 & 0 & a\end{array}\right),
\end{eqnarray}
where $0\le a \le 1$ and $a-1 \le b \le 1-a$. The discord of this state is \cite{AlQasimiPRA2011}
\begin{equation}\label{Q-Rho-ab}
Q(\rho^{AB}(a,b))=\min\{a, q\},
\end{equation}
where
\begin{eqnarray}\nonumber
q &=&\frac{a}{2}\log_{2}{\left[\frac{4a^2}{(1-a)^2-b^2}\right]}-\frac{b}{2}\log_{2}{\left[\frac{(1+b)(1-a-b)}{(1-b)(1-a+b)}\right]}\\ \nonumber
&-&\frac{\sqrt{a^2+b^2}}{2}\log_{2}{\left[\frac{1+\sqrt{a^2+b^2}}{1-\sqrt{a^2+b^2}}\right]} \\ \label{q}
&+&\frac{1}{2}\log_{2}{\left[\frac{4((1-a)^2-b^2)}{(1-b^2)(1-a^2-b^2)}\right]} .
\end{eqnarray}
For this state we get
\begin{eqnarray}
\vec{x}=-\vec{y}=\left(\begin{array}{c}0\\0\\-b\end{array}\right), \quad
T=\left(\begin{array}{ccc}a & 0 & 0 \\ 0 & -a & 0 \\0 & 0 & 2a-1\end{array}\right).
\end{eqnarray}
Clearly,  $T^{\T}\vec{x}=-(2a-1)\vec{y}$, so that for all values of $a$ and $b\ne 0$ we have  $\mathcal{R}=\textrm{span}{\{\vec{y}\}}$. Therefore an arbitrary element of $\mathcal{R}^{\perp}$ can be written as $\hat{e}_0=(\cos{\phi},\sin{\phi},0)^{\T}$. In this case we get $\hat{e}_0^{\T} T^{\T} T\hat{e}_0=a^2$, which is independent of $\phi$, so that  $t_0^2=a^2$. This means that every vector in the two-dimensional subspace $\mathcal{R}^{\perp}$, corresponding to the $xy$-plane,  gives the desired upper bound.  Therefore, we obtain
\begin{eqnarray}
\min_{\hat{e}_0\in \mathcal{R}^{\perp}}S(\rho^A|\{\hat{e}_0\})=h_2\left(\frac{1+\sqrt{a^2+b^2}}{2}\right),
\end{eqnarray}
and
\begin{eqnarray}
S(\rho^A)&=&S(\rho^B)=h_2\left(\frac{1+b}{2}\right), \\
S(\rho^{AB})&=&h_3\left(a,\frac{1-a-b}{2},\frac{1-a+b}{2}\right),
\end{eqnarray}
where $h_3(\mu_1,\mu_2,\mu_3)=-\sum_{i=1}^3 \mu_i\log_{2}{\mu_i}$. After some calculations we find  that the inequality (\ref{D-UpperBound}) leads to
\begin{equation}
Q(\rho^{AB}(a,b))\le q,
\end{equation}
where $q$ is defined in Eq. (\ref{q}). A comparison of this with Eq. (\ref{Q-Rho-ab}) shows that   for some region of parameters, namely when $q\le a$, our upper bound is tight. Figs. \ref{Fig-1} and \ref{Fig-2} illustrate this fact. These figures also reveal that for this class of states our bound is stronger than the upper bound introduced in Ref. \cite{XiJPA2011}.

\begin{figure}[h]
\epsfxsize=8.5cm \ \centerline{\hspace{0cm}\epsfbox{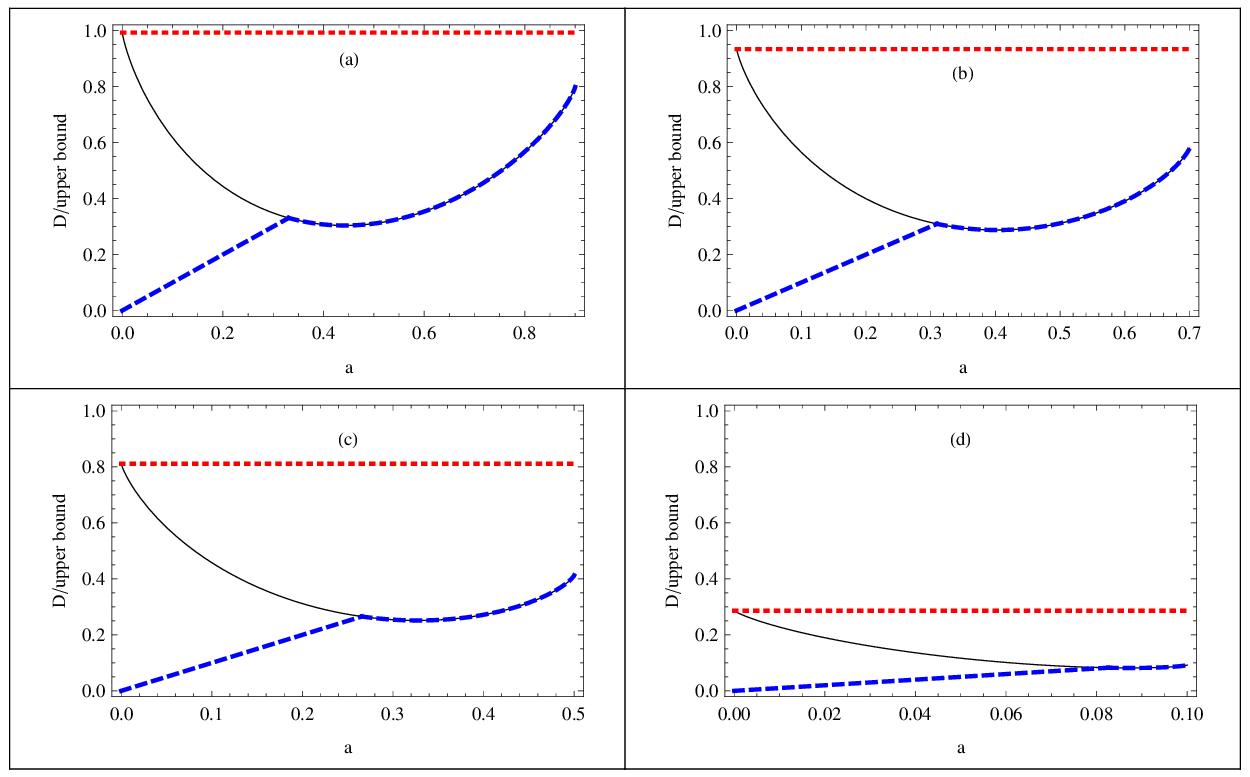}} \
\caption{(Color online) Quantum discord [dashed-blue lines], our upper bound [solid-black lines], and the upper bound of Ref. \cite{XiJPA2011}  [dotted-red lines] are plotted  versus $a$ for the state $\rho^{AB}(a,b)$ with:  (a) $b=0.1$, (b) $b=0.3$, (c) $b=0.5$, and  (d) $b=0.9$.}\label{Fig-1}
\end{figure}

\begin{figure}[h]
\epsfxsize=8.5cm \ \centerline{\hspace{0cm}\epsfbox{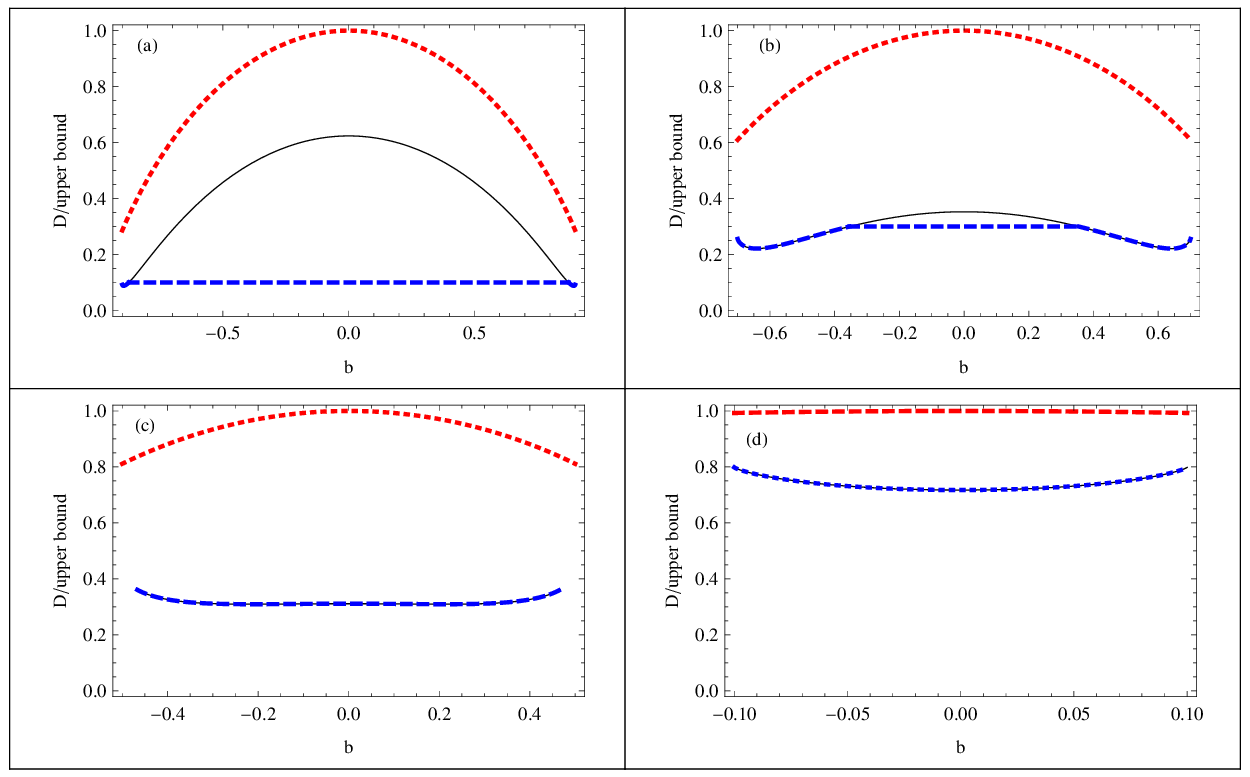}} \
\caption{(Color online) Quantum discord [dashed-blue lines], our upper bound [solid-black lines], and the upper bound of Ref. \cite{XiJPA2011}  [dotted-red lines] are plotted  versus $b$ for the state $\rho^{AB}(a,b)$  with: (a) $a=0.1$, (b) $a=0.3$, (c) $a=0.5$ (d) $a=0.9$. In the cases (c) and (d), our upper bound completely coincide with the quantum discord.}\label{Fig-2}
\end{figure}

\section{Optimization}
In this section we present an analytical procedure for optimization of the conditional entropy for  a general two-qubit state. We also  provide example for which one can obtain the minimum, analytically.
In order to determine the minimum of the conditional entropy (\ref{QC-Entropy-2}), we have to calculate its derivatives with respect to $\theta$ and $\phi$. To do this we need to calculate derivatives of the probabilities given by Eqs. (\ref{pk}) and (\ref{w1234}) with respect to $\theta$ and $\phi$. For instance, their derivative with respect to $\theta$ are as follows
\begin{eqnarray}
\frac{\partial p_{0,1}}{\partial \theta}&=&\pm \frac{1}{2}{\hat n}_{,\theta}^{\T} \;{\vec y},
\\
\frac{\partial w_{1,2}}{\partial \theta}&=&\frac{1}{4}{\hat n}_{,\theta}^{\T} \;\left[{\vec y}\pm T^{\T}{\hat Z^+}\right],
\\
\frac{\partial w_{3,4}}{\partial \theta}&=&\frac{1}{4}{\hat n}_{,\theta}^{\T} \;\left[{-\vec y}\pm T^{\T} {\hat Z^-}\right],
\end{eqnarray}
with
\begin{equation}
{\hat Z^+}=\frac{T{\hat n}+\vec{x}}{|T{\hat n}+\vec{x}|},\qquad {\hat Z^-}=\frac{T{\hat n}-\vec{x}}{|T{\hat n}-\vec{x}|},
\end{equation}
and the unit vector ${\hat n}_{,\theta}$ is defined by
\begin{equation}
{\hat n}_{,\theta}=\frac{\partial{\hat n}}{\partial \theta}=(\cos{\theta}\cos{\phi},\cos{\theta}\sin{\phi},-\sin{\theta})^{\T}.
\end{equation}
Evidently ${\hat n}\cdot{\hat n}_{,\theta}=0$.
By defining the nonunit vector ${\tilde n}_{,\phi}$ by
\begin{equation}
{\tilde n}_{,\phi}=\frac{\partial{\hat n}}{\partial \phi}=(-\sin{\theta}\sin{\phi},\sin{\theta}\cos{\phi},0)^{\T},
\end{equation}
 orthogonal to both ${\hat n}$ and ${\hat n}_{,\theta}$, we  get a similar equations for the derivatives of the probabilities with respect to $\phi$, but now ${\hat n}_{,\theta}$ is replaced by ${\tilde n}_{,\phi}$.
Finally using the above equations, we find the following relations for derivatives of the conditional entropy
\begin{eqnarray}\label{SC1}
\frac{\partial {S(\rho^A|\{\Pi_k^B\})}}{\partial\theta}&=&
-\frac{1}{4}{\hat n}_{,\theta}^{\T}\; \vec{A}, \\ \label{SC2}
 \frac{\partial {S(\rho^A|\{\Pi_k^B\})}}{\partial\phi}&=&
-\frac{1}{4}{\tilde n}_{,\phi}^{\T}\; \vec{A}
\end{eqnarray}
where $\vec{A}$ is a vector defined by
\begin{equation}\label{AVector}
\vec{A}=\left[\log_{2}{\frac{w_1w_2p_1^2}{w_3w_4p_0^2}}\right]\vec{y}+\left[\log_{2}{\frac{w_1}{w_2}}\right]T^{\T}{\hat Z^+}
+\left[\log_{2}{\frac{w_3}{w_4}}\right]T^{\T}{\hat Z^-}.
\end{equation}
Equations (\ref{SC1}) and (\ref{SC2}) enable one to present the following theorem as a necessary and sufficient condition for vector ${\hat n}=(\sin{\theta}\cos{\phi}, \sin{\theta}\sin{\phi}, \cos{\theta})^{\T}$   to be the stationary  measurement of  the conditional entropy, i.e. ${\partial {S(\rho^A|\{\Pi_k^B\})}}/{\partial\hat{n}}=0$, as follows
\begin{theorem}
Vector ${\hat n}\in\mathbb{R}^3$ is a stationary measurement  for the quantum conditional entropy if and only if vector $\vec{A}$, defined by Eq. (\ref{AVector}), satisfy the following condition
\begin{equation}\label{Stationary2}
{\hat n}_{\perp}^{\T} \;\vec{A}=0,
\end{equation}
where $\hat{n}_{\perp}$ is any vector perpendicular to $\hat{n}$, i.e. $\hat{n}^{\T}_{\perp}\hat{n}=0$.
\end{theorem}
\begin{remark}
Note that if we proceed the optimization progress  by using the Lagrange multiplier $\lambda$, having ${\hat n}^{\T}{\hat n}=1$ as a constraint, we find for the stationary condition $\dif\left[S(\rho^A|{\{\Pi_k^B}\})-\lambda({\hat n}^{\T}{\hat n}-1)\right]=0$ the following relation
\begin{equation}\label{Stationary3}
\vec {A}=A{\hat n},
\end{equation}
where  $\vec{A}$ is defined in Eq. (\ref{AVector}), and $A$ is given   by
\begin{eqnarray}\nonumber
A=&-&\left[\log_{2}{\frac{w_1}{w_2}}\right]\vec{x}^{\T}{\hat Z}^{+}+\left[\log_{2}{\frac{w_3}{w_4}}\right]\vec{x}^{\T}{\hat Z}^{-} \\
&-&4S(\rho^A|\{\Pi_k^B\})-\left[\log_{2}{\frac{w_1w_2w_3w_4}{p_0^2p_1^2}}\right].
\end{eqnarray}
\end{remark}
The stationary condition (\ref{Stationary3}) is equivalent to that given by Eq.  (\ref{Stationary2}) in the sense that both conditions require that in the extremum points,  vector $\vec{A}$ should be directed to  ${\hat n}$.
Unfortunately, these conditions do not have simple solutions, for ${\vec A}$  as well as $A$ depends  also on ${\hat n}$. Moreover, knowing the extremum points of the conditional entropy  is not enough to establish its minimum, and we need, in addition,  to evaluate the conditional entropy in the extremum points to find the minimum  one.
Below we exemplify these conditions for one particular class of states where we have already obtained the minimum of the conditional entropy without  optimization.

\noindent{\it States with $T^{\T}\vec{x}=0$ and $\vec{y}=0$.---}   For this class of states we have already shown that the optimal measurement lies  in the direction of the eigenvector of $T^{\T}T$, corresponding to the largest eigenvalue. We now reconsider this class of states and obtain the optimum measurement by using the optimization condition given above. For this class of states  vector ${\vec A}$ takes the following form
\begin{equation}
{\vec A}=\frac{2}{|T{\hat n}+\vec{x}|}\left[\log_{2}{\frac{1+|T{\hat n}+\vec{x}|}{1-|T{\hat n}+\vec{x}|}}\right]T^{\T}T{\hat n}.
\end{equation}
It is clear that the condition (\ref{Stationary2}) leads to ${\hat n}_{\perp}^{\T}T^{\T}T{\hat n}=0$, which has solutions  when ${\hat n}$ is an eigenvector of $T^{\T}T$. But the $3\times 3$ matrix $T^{\T}T$ has three nonnegative eigenvalues $\{t_1^2, t_2^2, t_3^2\}$ corresponding to the eigenvectors $\{\hat{e}_1,\hat{e}_2,\hat{e}_3\}$.
Therefore, in this particular case, three eigenvectors of $T^{\T}T$ are stationary measurements of the conditional entropy. For these directions, conditional entropy takes the following form
\begin{equation}
S(\rho^A|\{\hat{e}_k\})=h_2\left(\frac{1+\sqrt{x^2+t^2_k}}{2}\right), \quad \text{for} \quad k=1,2,3.
\end{equation}
Simple evaluation shows that the minimum of the above equation happens when $t^2_k$ corresponds to the largest eigenvalue of $T^{\T} T$.

Example above show that the stationary condition given in Eq. (\ref{Stationary2}) gives us, in general, more than one solution for the measurement direction $\hat{n}$, and we have to find  the optimal one by further evaluations. However,
the presented stationary condition is simple and computationally straightforward, in the sense that
it can be stored in a computer and that could be used for doing symbolic and numerical
calculations.

\section{Conclusion}
All difficulties in calculating the quantum discord arise from the difficulty in finding  the minimum of the quantum conditional entropy $S(\rho^{AB}|\{\Pi^B_{k}\})$ over all measurements on the subsystem  $B$. In this work, we have presented a simple relation for the quantum conditional entropy of a two-qubit system, as the difference of two Shannon entropies. Using  it,  we have obtained the quantum discord for  a class of states for which the conditions $T^{\T}\vec{x}=0$ and $\vec{y}=0$ are satisfied. This class of states includes the Bell-diagonal states, a three-parameter subclass of $X$ states,  and some zero-discord states.  Although, the quantum discord of these states is already obtained analytically, but the presented form for the conditional entropy enables one to obtain the previous results in a simpler manner.
 For these states, in particular, it is shown that the quantum conditional entropy reduces to the binary Shannon entropy, so that  in minimizing the conditional entropy we encountered  with the simple problem of minimizing binary Shannon entropy.

We have also shown that such obtained relation for the conditional entropy can be used to provide a computable tight upper bound on the quantum discord, so that the question of how large the quantum discord can possibly be, can be answered more reasonably.
We have presented sufficient conditions for the reachable upper bound of the quantum discord and have exemplified this bound for  a two-parameter class of states and have shown that the bound may be tight even in the absence of  sufficient conditions.

Based on the simple form of the conditional entropy, a general procedure of optimization of the conditional entropy is also given, and conditions under which conditional entropy is stationary are presented.  We show that our algorithm of optimization is efficient in the sense that it can be used to calculate quantum discord for some
classes of states, analytically.
The presented stationary condition is simple and computationally straightforward, in the sense that
it can be stored in a computer and that could be used for doing symbolic and numerical
calculations. The paper, therefore, can be
regarded as a further development in the calculation of the   quantum discord for an arbitrary state of two-qubit system.
The  method presented in this paper can be generalized to higher dimensional systems.

\newpage

\end{document}